\documentclass{endm}
\usepackage{endmmacro}
\usepackage{graphicx}

\usepackage{amsmath}
\usepackage{amssymb}
\usepackage{graphicx}
\usepackage{color}

\newcommand{\C}{\mathbb{C}}
\newcommand{\R}{\mathbb{R}}

\newcommand{\E}{\vec{E}}

\usepackage{varioref}        
\usepackage{fancyref}

\makeatletter
\def\mkfancyprefix#1#2{%
	\expandafter\def\csname fancyref#1labelprefix\endcsname{#1}%
	\begingroup\def\x{\endgroup\frefformat{plain}}%
	\expandafter\x\csname fancyref#1labelprefix\endcsname
	{\MakeLowercase{#2}\fancyrefdefaultspacing##1}%
	\begingroup\def\x{\endgroup\Frefformat{plain}}%
	\expandafter\x\csname fancyref#1labelprefix\endcsname
	{#2\fancyrefdefaultspacing##1}%
	\begingroup\def\x{\endgroup\frefformat{vario}}%
	\expandafter\x\csname fancyref#1labelprefix\endcsname
	{\MakeLowercase{#2}\fancyrefdefaultspacing##1##3}%
	\begingroup\def\x{\endgroup\Frefformat{vario}}%
	\expandafter\x\csname fancyref#1labelprefix\endcsname
	{#2\fancyrefdefaultspacing##1##3}%
}

\makeatletter
\newcommand{\removelatexerror}{\let\@latex@error\@gobble}
\makeatother
\renewcommand{\vec}[1]{\ensuremath{\mathbf{#1}}}
\newcommand{\N}{\mathbb{N}}

\newcommand{\CGab}{\mathcal{C}_\mathrm{G}}

\newcommand{\K}{K}
\renewcommand{\L}{L}
\newcommand{\LK}{\L/\K}
\newcommand{\Lset}{\L[x;\theta]}
\newcommand{\Gal}[1]{\mathrm{Gal}\left(#1\right)}
\newcommand{\GalLK}{\Gal{\L/\K}}

\newcommand{\Aa}{\mathcal{A}}
\newcommand{\Bb}{\mathcal{B}}
\newcommand{\X}{\mathbf{X}}
\newcommand{\x}{\mathbf{x}}
\newcommand{\Xo}{\mathbf{X_0}}
\newcommand{\xo}{\mathbf{x_0}}
\renewcommand{\H}{\mathbf{H}}
\newcommand{\y}{\mathbf{y}}
\newcommand{\ytilde}{\mathbf{\tilde{y}}}
\newcommand{\Y}{\mathbf{Y}}
\newcommand{\MSP}[1]{\mathcal{A}_{#1}}

\newcommand{\LH}[1]{\langle #1 \rangle}

\newcommand{\rk}{\mathrm{rank}}
\newcommand{\ext}{\mathrm{ext}}
\newcommand{\smallsum}{{\textstyle\sum\nolimits}}

\newcommand{\Fnorm}[1]{\|#1\|_\mathrm{F}}

\begin{document}

\begin{verbatim}\end{verbatim}\vspace{2.5cm}

\begin{frontmatter}

\title{Low-Rank Matrix Recovery using Gabidulin Codes in Characteristic Zero}

\author{Sven M\"uelich\thanksref{ALL}\thanksref{sven.mueelich@uni-ulm.de}, Sven Puchinger, Martin Bossert}
\address{Institute of Communications Engineering\\ University of Ulm\\ Ulm, Germany}

\thanks[ALL]{This work has been supported by DFG, Germany, under grant BO 867/32-1} \thanks[sven.mueelich@uni-ulm.de]{Email:
   \href{mailto:sven.mueelich@uni-ulm.de} {\texttt{\normalshape
   {sven.mueelich,sven.puchinger,martin.bossert}@uni-ulm.de}}} 

\begin{abstract}
We present a new approach on low-rank matrix recovery (LRMR) based on Gabidulin Codes.
Since most applications of LRMR deal with matrices over infinite fields, we use the recently introduced generalization of Gabidulin codes to fields of characterstic zero. We show that LRMR can be reduced to decoding of Gabidulin codes and discuss which field extensions can be used in the code construction.
\end{abstract}

\begin{keyword}
Gabidulin Codes, Characteristic Zero, Low-Rank Matrix Recovery
\end{keyword}

\end{frontmatter}

\section{Introduction}
Low-rank matrices occur in many applications, e.g., in signal theory, machine learning and collaborative filtering. Unfortunately, in many cases it is only possible to get incomplete or indirect information of matrices. Since applications usually require complete matrices in order to process data, it is necessary to recover matrices from available data. In general this is not possible. However, when matrices are of low-rank, there are efficient algorithms to accomplish this task. So far, \emph{low-rank matrix recovery (LRMR)} is described by a minimization problem which can be solved by convex optimization programs. 
In this work we show how Gabidulin codes can be used in order to solve the LRMR problem. 

\section{Low-Rank Matrix Recovery}
	\label{sec:lrmr}
\emph{Low-Rank Matrix Recovery (LRMR)} was first defined in \cite{candes2009exact,candes2010power,recht2010guaranteed} and can be seen as matrix-analogue of compressed sensing \cite{gross2011recovering}.
The goal of LRMR is to reconstruct a matrix from incomplete or indirect observations. The problem is stated as follows: We want to recover an unknown matrix $\Xo \in \K^{m \times n}$ of lowest possible rank, where in applications $\K$ usually is the real or complex field. Therefore, we use observed measurements $\y = \Aa(\Xo)$, which we obtain by applying a linear measurement operator $\Aa:\R^{m \times n} \rightarrow \R^{p}$ to $\Xo$. We assume that $\Aa$ can be chosen arbitrarily. Finding a solution to this problem can be specified in terms of the minimization problem 
\begin{align}
	\min  \rk(\mathbf{\X}) ~\text{subject to }  \Aa(\X) =  \Aa(\Xo).
	\label{EqMinimization}
\end{align} 
Often, in literature (\ref{EqMinimization}) is also called rank minimization problem.
Since this problem is NP-hard, convex relaxations are considered. Most commonly used algorithms are nuclear norm minimization \cite{recht2010guaranteed} and iterative hard thresholding \cite{jain2010guaranteed}. 
An overview of these and other methods, theoretical guarantees and applications is given in \cite{davenport2016}.  
		
\section{Gabidulin Codes in Characteristic Zero}
\label{sec:gabcodes}
Gabidulin codes over finite fields were introduced in \cite{delsarte1978bilinear,gabidulin1985theory,roth1991maximum}, a comprehensive overview is given in \cite{wachter2013decoding}. Since we deal with numbers from infinite alphabets in LRMR, there is a need for Gabidulin codes in characteristic zero, which we introduce according to \cite{augot2013rank,robert2015phd,muelich2016alternative}.
Decoding in rank metric can be described by $\min  \rk(\mathbf{\E'}) ~\text{subject to }  \H\mathbf{\E'} =  \H\E$,
which has a similar form as Equation~\eqref{EqMinimization}. This observation suggests that it might be possible to use a rank metric decoder in order to recover a low-rank matrix and in Section~\ref{sec:newapproach} we will show how this can be done.
	
Let $\K\subseteq\L$ be fields and $\LK$ be a field extension of degree $m$. A codeword of a Gabidulin code can either be an $(m \times n)$-matrix over the ground field $K$ or a vector of length $n$ over $\L$. Let $\Bb = \{\beta_0,\dots,\beta_{m-1}\}$ be a basis of $\L$ over $\K$, $\boldsymbol{\beta}= \{\beta_0,\dots,\beta_{m-1}\}$ an order of the basis and  $\x \in \L$. We use the following bijective mapping in order to switch between the two representations:
\begin{definition}
The $\ext$-mapping is a bijective function $\ext_{\beta}:\L^n \rightarrow \K^{m\times n}$ related to an ordered basis $\boldsymbol{\beta}$ of $\L$ which fulfills the equation 
\begin{align}
\label{eq:ext}
\x = \beta\cdot \ext_{\beta}(\x) = \beta \cdot \X. 
\end{align}
Equation~(\ref{eq:ext}) contains the inverse of the mapping. The $\boldsymbol{\beta}$ in the notation of the mapping $\ext_{\beta}(\x)$ is omitted if it is not important to the problem.
\end{definition}
	
Gabidulin codes in characteristic zero are defined using $\theta$-polynomials which are a generalization of linearized polynomials known from the definition of Gabidulin codes over finite fields.
\begin{definition}
Let $\K\subseteq\L$ be fields and $\LK$ be a Galois extension. The Galois group of $\LK$ is $\GalLK = \{ \theta: \L \rightarrow \L$ automorphism $|~\theta(k) = k~\forall k \in \K\}$.	 
 The set of $\theta$-polynomials is defined as 
\begin{align*}
	\Lset = \left\{a = \smallsum_{i=0}^{d_a} a_i x^i : a_i \in \L, \, d_a \in \N, \, a_{d_a} \neq 0 \right\}.
\end{align*}
\end{definition}
$\Lset$ with ordinary addition and multiplication $x\cdot\alpha = \theta(\alpha)$  $\forall \alpha \in \L$ extended to polynomials inductively, is a ring, the so-called $\theta$-polynomial ring. The degree of $a \in \Lset$ is given by $\deg a = d_a$ and
$a$ is called \emph{monic} if $a_{d_a} = 1$.
Let $\K \subseteq \L$ be fields and $\L/\K$ be a Galois extension of degree $m$. We denote the number of linearly independent columns over $\K$ is by $\rk_\K$. We define
	\arraycolsep=0pt
\begin{align*}
	\mathbf{X_{\theta}} = 
	\begin{pmatrix}
		x_{1} & \dots & x_{n} \\
		\vdots & \ddots & \vdots \\
		\theta^{s-1}(x_{1}) & \dots &  \theta^{s-1}(x_{n})
	\end{pmatrix} \text{and~}
	\mathbf{X_{\Bb}} = 
	\begin{pmatrix}
		x_{1,1} & \dots & x_{n,1} \\
		\vdots & \ddots & \vdots \\
		x_{1,m} & \dots &  x_{n,m}
	\end{pmatrix}.
\end{align*}

\cite[Section~2.2]{robert2015phd} gives four definitions of rank weight in characteristic zero. 
\begin{definition}\label{def:rank_metrics}
	Let $\mathbf{x} \in \L^n$. The following are rank weights:
	\begin{align*}
		\omega_{\Aa}(\mathbf{x}) = deg(\MSP{\LH{x_1,\dots,x_n}}),~~
		\omega_{\theta,\L}(\mathbf{x}) = \rk_\L(\mathbf{X_{\theta}}),\\
		\omega_{\theta,\K}(\mathbf{x}) = \rk_\K(\mathbf{X_{\theta}}),~~
		\omega_{\Bb}(\mathbf{x}) = \rk_\K(\mathbf{X_{\Bb}}),
	\end{align*}
where $\MSP{\LH{x_1,\dots,x_n}}$ is the Minimal Subspace Polynomial as defined in \cite{robert2015phd}. The corresponding rank metrics can be defined as 
$\mathrm{d}_\mathrm{R}(\mathbf{x},\mathbf{y}) = \omega_i (\mathbf{x}-\mathbf{y})$.
\end{definition}
	
In the finite field case, these rank weights are the same.
In {\cite[Lemmata~13, 14, and 15]{robert2015phd}}, the relation $\omega_{\Aa}(\mathbf{x}) = \omega_{\theta,\L}(\mathbf{x}) \leq \omega_{\theta,\K}(\mathbf{x}) = \omega_{\Bb}(\mathbf{x})$ has been proven over characteristic zero.
	
\begin{definition}
Let $g_1,\dots,g_n \in \L$ be linearly independent over $\K$. Then a Gabidulin code with parameters $n$ and $k \leq n$ is defined as
\begin{align*}
	\CGab[n,k] = \left\{ (f(g_1), \dots, f(g_n)) \, : \, f \in \Lset \, \land \, \deg f < k  \right\}.
\end{align*}
\end{definition}
For an overview of properties we refer to \cite{robert2015phd}.
For decoding Gabidulin codes in characteristic zero, \cite{augot2013rank} gives a generalization of a Welch-Berlekamp-like algorithm which allows decoding in $O(n^3)$. In \cite{muelich2016alternative} we derived a Gao-like key equation and hence reduced the decoding complexity to $O(n^2)$.
	
\section{New Approach for Low-Rank Matrix Recovery}
\label{sec:newapproach}
In this section we reduce the problem of LRMR to decoding of Gabidulin codes and show how decoding can be used in order to recover low-rank matrices. 
Let $\H \in \L^{(n-k)\times n}$ be the parity-check matrix of a Gabidulin code.
Recalling Equation~(\ref{EqMinimization}), we interpret the unknown matrix $\Xo \in \K^{m \times n}$ as the low-rank error $\E$ which usually arises while transmitting a codeword of a Gabidulin code. We define the linear measurement operator $\Aa:\K^{m \times n} \rightarrow \K^{p}$ as in Algorithm~\ref{alg:measurementop}.

\begin{alg}
	Linear measurement operator \\
	Input: $\Xo \in \K^{m \times n}$ \\
	Output: $y \in \K^p$ \\
	1. $\xo \gets \ext^{-1}(\Xo)$  ~~~~~~~~~~~~~~~~~~~~~~~~~~~~~// {$\xo \in L^{n}$}   \label{line:lmo_1} \\
	2. $\ytilde \gets \H \cdot \xo^{T}$ ~~~~~~~~~~~~~~~~~~~~~~~~~~~~~~~~~// {$\ytilde \in \L^{n-k}$}  \label{line:lmo_2} \\ 
	3. $\Y \gets \ext(\ytilde)$ ~~~~~~~~~~~~~~~~~~~~~~~~~~~~~~~~~// {$\Y \in \K^{(n-k)\times n}$} \label{line:lmo_3}\\
	4. $\y \gets$ vector representation of $\Y$ ~~~~~~~~~// {$\y \in \K^{n\cdot (n-k)}$} \label{line:lmo_4}\\
	5. Return $\y$
	\label{alg:measurementop}
\end{alg}

Note, that all operations used in the algorithm are $\K$-linear. 

\begin{theorem}
If $p = n \cdot (n-k)$, there exists a mapping $\Aa$ as in Algorithm~\ref{alg:measurementop} such that $\Xo$ with $\rk(\Xo) \leq \frac{n-k}{2} = \frac{p}{2n}$ can be recovered.
\end{theorem}
	
\begin{proof}
We use the output of the linear measurement operator to calculate $\ytilde~=~\ext^{-1}(y)~\in~\L^{n-k}$. This corresponds to the syndrome which is used by a decoder in order to produce the error matrix $\E$ which is in our case the matrix $\Xo$. We know from coding theory, that if $\rk(\Xo) \leq \frac{d-1}{2} = \frac{n-k}{2} = \frac{p}{2n}$ the matrix $\Xo$ (and consequently $\E$) is unique and can be found by a bounded minimum distance decoder \cite{muelich2016alternative}.
\end{proof}
		
Using our result from \cite{muelich2016alternative}, decoding can be done in $O(n^2)$. 
In comparison, the complexity of previous methods depends on singular value decomposition.
	
\section{Which Fields Can Be Used?}
\label{sec:fields}
Applications usually use real or complex matrices. Using Gabidulin codes for solving LRMR, we cannot use $\K = \C$, since $\C$ is algebraically closed. It is well-known that the only finite extensions of $\R$ are $\C$ or $\R$ and hence $\K = \R$ is also not suitable. Thus, we need another field $K$ which approximates the actual matrix. We can find such a $\K$, e.g., by minimizing the Frobenius norm of the difference matrix: Assuming that $\X \in \C^{m \times n}$ is the actual matrix, we would like to choose a field $K$ such that we can find another matrix $\X' \in K^{m \times n}$ with
\begin{align*}
	\Fnorm{\X'-\X} = \sqrt{\sum\limits_{i,j} |x'_{ij} - x_{ij}|^2} < \epsilon
\end{align*}
for any $\epsilon>0$. This condition is always fulfilled if $K$ is a \emph{dense} subfield of $\C$ and helps to minimize the ``approximation error'' when going from $\C$ to $K$.
	
If we want to recover a matrix over real numbers, we propose to use cyclotomic extensions. 
Therefore, we take the rational numbers and adjoint the $n$-th roots of unity. Hence, we get a field extension of degree $\varphi(n)$, where $\phi$ is Euler's phi function. We are able to construct a cyclotomic extension of degree $n$ iff $n \in \textrm{image}(\varphi)$. Since for every prime $p$ it is $p-1 \in \textrm{image}(\varphi)$, the density of $\textrm{image}(\varphi)$ is at least $\frac{n}{\textrm{ln}(n)}$ (cmp. prime number theorem). Hence, $\textrm{image}(\varphi)$ is sufficiently dense to find suitable numbers. The automorphism group of such a field fulfills the required properties, namely the Galois group is cyclic and the characteristic polynomial of an automorphism $\theta \in Gal$ is square-free. Examples can be found in {\cite[Section~2.3.3]{robert2015phd}}.
	
If the desired matrix is in $\mathbb{C}$, we can use Kummer extensions {\cite[Section~2.3.3]{robert2015phd}}. We need to choose $\K$ to be a dense subfield of $\mathbb{C}$ which has $n$ many $n$-th roots of unity. We can ensure this by choosing $\K=\mathbb{Q}(\zeta_n)$, where $\zeta_n = \mathrm{e}^{\mathrm{i} \frac{2 \pi}{n}}$, i.e., $\K$ is a cyclotomic extension of $\mathbb{Q}$ of degree $\varphi(n)$. Then, $\K \subseteq \mathbb{C}$ dense and it contains $n$ distinct $n$-th roots of unity (namely $\zeta_n^i$ for $i=0,\dots,n-1$).
Thus, there is a Kummer extension $\L/\K$ of degree $n$, e.g., by adjoining an $n$-th root of $2$ to $\K$.
This approach works for any $n\in \mathbb{N}$ with $n>2$.


\begin{thebibliography}{10}
	\providecommand{\url}[1]{#1}
	\csname url@samestyle\endcsname
	\providecommand{\newblock}{\relax}
	\providecommand{\bibinfo}[2]{#2}
	\providecommand{\BIBentrySTDinterwordspacing}{\spaceskip=0pt\relax}
	\providecommand{\BIBentryALTinterwordstretchfactor}{4}
	\providecommand{\BIBentryALTinterwordspacing}{\spaceskip=\fontdimen2\font plus
		\BIBentryALTinterwordstretchfactor\fontdimen3\font minus
		\fontdimen4\font\relax}
	\providecommand{\BIBforeignlanguage}[2]{{%
			\expandafter\ifx\csname l@#1\endcsname\relax
			\typeout{** WARNING: IEEEtran.bst: No hyphenation pattern has been}%
			\typeout{** loaded for the language `#1'. Using the pattern for}%
			\typeout{** the default language instead.}%
			\else
			\language=\csname l@#1\endcsname
			\fi
			#2}}
	\providecommand{\BIBdecl}{\relax}
	\BIBdecl
	
	\bibitem{candes2009exact}
	E.~J. Cand{\`e}s and B.~Recht, ``{Exact Matrix Completion via Convex
		Optimization},'' \emph{Foundations of Computational Mathematics}, vol.~9,
	no.~6, pp. 717--772, 2009.
	
	\bibitem{candes2010power}
	E.~J. Cand{\`e}s and T.~Tao, ``{The Power of Convex Relaxation: Near-optimal
		Matrix Completion},'' \emph{IEEE Transactions on Information Theory},
	vol.~56, no.~5, pp. 2053--2080, 2010.
	
	\bibitem{recht2010guaranteed}
	B.~Recht, M.~Fazel, and P.~A. Parrilo, ``{Guaranteed Minimum-Rank Solutions of
		Linear Matrix Equations via Nuclear Norm Minimization},'' \emph{SIAM review},
	vol.~52, no.~3, pp. 471--501, 2010.
	
	\bibitem{gross2011recovering}
	D.~Gross, ``{Recovering Low-Rank Matrices from few Coefficients in any
		Basis},'' \emph{IEEE Transactions on Information Theory}, vol.~57, no.~3, pp.
	1548--1566, 2011.
	
	\bibitem{jain2010guaranteed}
	P.~Jain, R.~Meka, and I.~S. Dhillon, ``{Guaranteed Rank Minimization via
		Singular Value Projection},'' in \emph{Advances in Neural Information
		Processing Systems}, 2010, pp. 937--945.
	
	\bibitem{davenport2016}
	M.~Davenport and J.~Romberg, ``An Overview of Low-Rank Matrix Recovery from
		Incomplete Observations,'', \emph{arXiv preprint 1601.06422}, 2016
	
	\bibitem{delsarte1978bilinear}
	P.~Delsarte, ``{Bilinear Forms over a Finite Field, with Applications to Coding
		Theory},'' \emph{Journal of Combinatorial Theory, Series A}, vol.~25, no.~3,
	pp. 226--241, 1978.
	
	\bibitem{gabidulin1985theory}
	E.~M. Gabidulin, ``{Theory of Codes with Maximum Rank Distance},''
	\emph{Problemy Peredachi Informatsii}, vol.~21, no.~1, pp. 3--16, 1985.
	
	\bibitem{roth1991maximum}
	R.~M. Roth, ``{Maximum-Rank Array Codes and their Application to Crisscross
		Error Correction},'' \emph{IEEE Transactions on Information Theory}, vol.~37,
	no.~2, pp. 328--336, 1991.
	
	\bibitem{wachter2013decoding}
	A.~Wachter-Zeh, ``{Decoding of Block and Convolutional Codes in Rank Metric},''
	Ph.D. dissertation, Universit{\'e} Rennes 1, 2013.
	
	\bibitem{augot2013rank}
	D.~Augot, P.~Loidreau, and G.~Robert, ``{Rank Metric and Gabidulin Codes in
		Characteristic Zero},'' in \emph{ISIT 2013}, 2013.
	
	\bibitem{robert2015phd}
	G.~Robert, ``{Codes de Gabidulin en Caract\'{e}ristique Nulle. Application au
		Codage Espace-Temps},'' Ph.D. dissertation, Universit\'{e} Rennes 1, 2015.
	
	\bibitem{muelich2016alternative}
	S.~M{\"u}elich, S.~Puchinger, D.~M{\"o}dinger, and M.~Bossert, ``{An
		Alternative Decoding Method for Gabidulin Codes in Characteristic Zero},''
	\emph{ISIT 2016}, 2016.
	
\end{thebibliography}
\end{document}